\documentclass[11pt]{article}

\usepackage[margin=1.3in]{geometry}
\usepackage{amsmath}
\usepackage{amssymb}
\usepackage{fullpage}
\usepackage{color}
\usepackage{graphicx}
\usepackage{epsfig}
\usepackage{amsthm}
\usepackage{latexsym}
\usepackage{amssymb}
\usepackage{amsmath}
\usepackage{verbatim}
\usepackage{graphicx}
\usepackage{epsfig}
\usepackage{cite}
\usepackage{url}
\usepackage[table]{xcolor}
\usepackage{tabularx}
\usepackage{algorithmic}
\usepackage{algorithm}
\usepackage{subfig}
\DeclareMathOperator*{\argmax}{arg\,max}

\newtheorem{theorem}{Theorem}
\newtheorem{proposition}{Proposition}

\newtheorem{lemma}{Lemma}
\newtheorem{corollary}{Corollary}
\newtheorem{definition}{Definition} 

\newcommand{\rank}{\mbox{rank}}
\newcommand{\Supp}{\textrm{Supp}}

\newcommand{\D}{\mathcal{D}}
\newcommand{\BR}{\beta}

\newcommand{\Ob}{\mathcal{O}}
\newcommand{\dm}{\textrm{dim}}
\newcommand{\rcn}{\kappa_r}
\newcommand{\ccn}{\kappa_c}
\newcommand{\cn}{\kappa}
\newcommand{\comp}{\sigma}

\begin{document}
\title{On the Existence of Low-Rank Explanations for \\  Mixed Strategy Behavior}
\author{
Siddharth Barman\thanks{Center for the Mathematics of Information, California Institute of Technology. {\tt barman@caltech.edu}}
\and Umang Bhaskar \thanks{Center for the Mathematics of Information, California Institute of Technology. {\tt umang@caltech.edu}}
\and Federico Echenique \thanks{Humanities and Social Sciences, California Institute of Technology. {\tt fede@hss.caltech.edu}}
\and Adam Wierman \thanks{Computing and Mathematical Sciences, California Institute of Technology. {\tt adamw@caltech.edu}}}

\date{}



\maketitle

\begin{abstract}
Nash equilibrium is used as a model to explain the observed behavior of players in strategic settings. For example, in many empirical applications we observe player behavior, and the problem is to determine if there exist payoffs for the players for which the  equilibrium corresponds to observed player behavior. Computational complexity of Nash equilibria is an important consideration in this framework. If the instance of the model that explains observed player behavior requires players to have solved a computationally hard problem, then the explanation provided is questionable. In this paper we provide conditions under which Nash equilibrium is a reasonable explanation for strategic behavior, i.e., conditions under which observed behavior of players can be explained by games in which Nash equilibria are easy to compute. We identify three structural conditions and show that if the data set of observed behavior satisfies any of these conditions, then it is consistent with payoff matrices for which the observed Nash equilibria could have been computed efficiently. Our conditions admit large and structurally complex data sets of observed behavior, showing that even with complexity considerations, Nash equilibrium is often a reasonable model. 


\end{abstract}






\section{Introduction}
\label{s.intro}

The computational complexity of equilibria in economic models is at the
core of recent research in algorithmic game theory. In general, the basic
message of this research is negative: Computing Nash equilibria is
PPAD-complete even for 2-player games~\cite{ChenDT09}, and
computing Walrasian equilibria is
PPAD-hard~\cite{ChenDDT09,VaziraniY11}. The fact that the standard notions of
equilibrium used by economists are hard to compute raises a concern that they are  flawed models for
economic behavior, because they lack a plausible rationale for how one would
arrive at equilibrium.

These hardness results have motivated the study of instances in which equilibrium can be computed efficiently, e.g.,~\cite{jain2007polynomial, Adsul2011, GargJM11}. These results --- both positive and negative ---  assume a fixed and literal interpretation of economic models. For example, in two-player games the assumption is that payoff matrices are explicitly specified, and agents seek to maximize their payoffs given the strategy of the other player.

However, models are frequently used as explanations of behavior, rather than literal descriptions. In empirical applications of game theory, for example, one rarely observes payoffs of the players. Rather, what we observe is the behavior of the players, and a first-order task is to determine if the model explains observed behavior or not (see, e.g., \cite{berry, bresnahan1989empirical} for classical examples of this exercise). In particular, a good model instance (i.e., payoffs specification) is one where the observed behavior of the players is as if they were playing equilibrium. This is certainly how economists use game theory: as a modeling tool to explain observed behavior. 


This perspective however cannot ignore computational complexity. Even if a model explains observed behavior, it is unreasonable if it requires players to solve computationally hard problems. Thus, an important question that arises from the perspective of models as explanations of observed behavior and  the computational complexity of equilibrium is: \emph{when is an economic model a reasonable explanation for observed behavior?} We address this question in the context of Nash equilibrium as a model for two-player games, and computational complexity as a measure of reasonableness.


This framework of starting with observed behavior is consistent with much of economics, and is formalized in \emph{revealed preference theory}, which was pioneered by Samuelson in 1938~\cite{samuelson1938note}
and has a long tradition in economics (see, e.g.,
\cite{afria67,varia82,varia83,varia84,varian2006revealed}). Classical revealed preference theory asks, ``Does there exist an instance of the model that is consistent with the observed behavior?'' If there is such a model instance, then the observed behavior is said to be consistent with the theory, and thus \emph{rationalizable}.  
Our work augments the fundamental question of rationalizability in revealed preference theory with
complexity considerations. 
The question we address is then, ``Does there exist an instance of the model that is consistent with the observed behavior \emph{and} for which the observations could have been computed efficiently?''

We focus on addressing this question in the context of bimatrix games.  We consider a setting where mixed strategy behavior (i.e., probability distributions over the actions of the players) is observed. In particular, with multiple instances of mixed strategy behavior as input, we want to know if there is some tractable instance of the game such that the given observations correspond to Nash equilibria. We focus on low \emph{player rank} --- the minimum of the ranks of the payoff matrices of the players --- as our notion of tractability given recent algorithmic results, e.g., \cite{GargJM11}. We defer a discussion of our modeling assumptions and their relaxations to Section~\ref{s.end}. 


It is important to be clear about what we do not do. We do not address a problem of inference; we do not claim to estimate or back out any instance of the game. Our focus is on the problem of {\em testing}: When can a dataset be explained by equilibrium behavior for which there is a story (an efficient algorithm) for how players could have arrived at the equilibrium? The existence of such an
explanation does not imply that the explanation is unique; indeed
there may be other explanations with different properties, a common
situation in revealed preference theory. Our results say that one can quite
often find explanations for economic behavior which do not require agents to solve computationally hard problems. Hence in such settings Nash equilibria, even with complexity considerations, remains an applicable model of behavior.

\paragraph{Summary of results.} Broadly, the results in this paper show that large and structurally complex data sets of observations --- including data sets with overlapping observed strategies, each of which may be an arbitrary distribution --- can be rationalized by games that  admit efficient computation of the observed Nash equilibria. 
Specifically, we identify three measures of structural
complexity in data sets, and show that if a data set has a low value
on any one of these measures for either player, or can be partitioned in a manner such that each partition has a low value on any one of these measures, then it has a rationalization with low player rank.  The three measures we study are (i) the
\emph{dimensionality} of the observed strategies (Theorem~\ref{thm:ub-count}), (ii) the
\emph{support size} of the observed strategies (Theorem~\ref{thm:ss}), and (iii) the
\emph{chromatic number} of the data set (Theorem~\ref{thm:cn}). We believe these are
natural and complementary measures to evaluate the structural
complexity of a data set, and our contribution is to show that each of
these measures individually translates to the existence of rationalizing games with low player rank.

Finally, we also show that the bounds we obtain on the player rank for these measures are nearly tight by giving an example of a data set that necessitates high player rank for any rationalization (Theorem \ref{thm:lb}). 

\paragraph{Related work.} 
This paper seeks to incorporate computational complexity as a measure of reasonableness of models (which are used to explain observed behavior) in the context of bimatrix games. This general direction was initiated in the context of consumer choice theory \cite{consumer}. In this context it was shown that a data set of $n$ observations of a consumer choosing among $d$ indivisible goods can always be explained by a consumer utility function that can be optimized in $O(nd)$ time.  Thus, despite the fact that the consumer choice problem is NP-hard, consumer choice data is not rich enough to expose computationally hard utility functions without an exponentially large data set.

The question was also asked in the context of bimatrix games in \cite{BarmanBEW13}; however the focus in that work was on {\em pure strategy equilibria} and so differs from the current paper, which focuses on mixed strategy equilibria.  In \cite{BarmanBEW13}, the core problem is not computationally hard, and so the focus is on the structural complexity of the rationalizing game, as formalized via game rank.\footnote{Game rank is defined to be the rank of the sum of the payoff matrices of the players. On the other hand, player rank is the minimum of the ranks of the payoff matrices of the players.} In contrast, the purpose of the current paper is to develop a connection between the structural richness of the data and the {\em computational} complexity of the game. 

This problem at hand is significantly more difficult and technically involved than the ones solved in \cite{consumer} and  \cite{BarmanBEW13}. To highlight this, note that, there are no results in the revealed preference literature characterizing which data sets can be explained as mixed Nash equilibria.  So, our results
represent a contribution to pure economic theory, as well as to
algorithmic game theory. 


\section{Preliminaries}
\label{s.prelims}

\textbf{Bimatrix Games.} Bimatrix games are two player games in normal
form.  Such games are specified by a pair of matrices $(A, B)$ of size
$n \times n$, which are termed the payoff matrices for the
players. The first player, also called the row player, has payoff
matrix $A$, and the second player, or the column player, has payoff
matrix $B$. The strategy set for each player is
$[n]=\{1,2,\ldots,n\}$, and, if the row player plays strategy $i$ and
column player plays strategy $j$, then the payoffs of the two players
are $A_{ij}$ and $B_{ij}$ respectively. The player rank of game
$(A,B)$ is defined to be $\min \{ \rank(A), \rank(B) \}$. Our focus on player rank stems from a number of important properties.  In particular, if a game has player rank $k \in \{1, 2, \dots, n\}$, then an equilibrium can be computed in time $O(n^{O(k)})$~\cite{GargJM11,LMM}. 

Let $\Delta^n$ be the set of probability distributions over the set of pure strategies $[n]$. For $x \in \Delta^n$, we define $\Supp(x) := \{i: x_i > 0\}$. Further, $e_i \in \mathbb{R}^n$ is the vector with $1$ in the $i$th coordinate and $0$'s elsewhere, and $u_k \in \Delta^n$ is the uniform distribution over the set $\{1,2, \ldots, k \}$. The players can randomize over their strategies by selecting any probability distribution in $\Delta^n$, called a mixed strategy. When the row and column players play mixed strategies $x$ and $y$ respectively, the expected payoff of the row player is $x^T A y $ and the expected payoff of the column player is $x^TBy$.

Given a mixed strategy $y \in \Delta^n$ for the column player, the
best-response set of the row player, $\BR_r$, is defined as $\BR_r(y)
:= \{ i \in [n] \mid e_i^T A y \geq e_k^T A y \ \ \forall k\in [n]
\}$. Similarly, the
best-response set, $\BR_c$, of the column player (against mixed
strategy $x \in \Delta^n$ of the row player) is defined as $\BR_c(x)
:= \{ j \in [n] \mid x^T B e_j \geq x^T B e_k \ \  \forall k \in [n]\} $.  The best
response sets $\BR_r$ and $\BR_c$ are defined with respect to the
payoff matrices $A$ and $B$. When we want to emphasize this fact we
use superscripts: $\BR_r^A$ and $\BR_c^B$.



\begin{definition}[Nash Equilibrium]
A pair of mixed strategies $(x,y)$, $x, y \in \Delta^n$, is a Nash equilibrium if and only if:
\begin{align*}
x^T A y & \geq e_i^T A y \qquad \forall i \in [n] \textrm{ and }\\
x^T B y & \geq x^T B e_j \qquad \forall j \in [n].
\end{align*}
The Nash equilibrium is \emph{strict} if additionally the support and the best-response sets are equal, i.e., $\Supp(x) = \BR_r(y)$ and $\Supp(y) = \BR_c(x)$.  Thus, for strict Nash equilibrium $x^T A y > e_i^T A y$ for all $ i \notin Supp(x)$ and $x^T B y > x^TB e_j $ for all $j \notin \Supp(y)$.
\end{definition}

\noindent \textbf{Observed behavior.}
We consider settings in which the payoff matrices $A$ and $B$ are not explicitly
specified. Instead, we are given a
collection of observed mixed strategy pairs. Our framework of observing mixed
strategies is entirely analogous to the theory of individual
stochastic choice: see
for example
\cite{luce1959,mcfadden_zarembka,mcfad91,mcfadden2005revealed}. The
idea is  that repeated observation of pure play allows one to infer a
probability distribution over pure strategies.




A data set $\D$ (of size $m$) is a
collection of mixed-strategy pairs, $\D=\{(x_k, y_k) \in \Delta^n
\times \Delta^n \mid k \in [m] \}$. We refer to mixed-strategy pairs
$(x_k, y_k)$ as observations.


We denote the set of observed mixed strategies of the row and column
player in a data set $\D$ by $\Ob_r(\D) $ and $\Ob_c(\D)$
respectively: $\Ob_r(\D) := \{ x \in \Delta^n \mid \exists  y \in
\Delta^n  \textrm{ such that } (x,y) \in \D\}$ and $\Ob_c(\D) := \{ y
\in \Delta^n \mid \exists  x \in \Delta^n  \textrm{ such that } (x,y)
\in \D\}$. When there is a single data set under consideration, for
ease of notation, we simply refer to these sets as $\Ob_r$ and
$\Ob_c$.

\section{Warmup result: Rationalization}

Given observed behavior (data) as described in the previous section, the first-order goal of classical revealed preference theory is to understand whether the data is rationalizable. 


\begin{definition}[Rationalizable Data]
A data set $\D=\{ (x_k, y_k)\}_k$ is  \emph{rationalizable} if there exist payoff matrices $A$ and $B$ such that for all $k$, $(x_k, y_k)$ is a strict Nash equilibrium in the game $(A,B)$.
\end{definition}

\noindent Strictness is required in the above definition in order to avoid rationalization by trivial games.


The first result of this paper comes from the observation that rationalizability of a given data set can be expressed as a linear program, and can thus be determined efficiently. In fact, the use of a linear program is a robust tool, and allows us to determine rationalizability in many more settings, such as:

\begin{itemize}
\item Rationalizability by payoff matrices that satisfy given linear constraints, e.g., enforcing values for certain payoffs, or inequalities between payoffs.
\item Games with multiple players. Each player now has a $m$-dimensional payoff tensor, where $m$ is the number of players.
\item Data sets that consist of observed player behavior in subgames. Each observation thus consists of a tuple $(I_k, J_k, x_k, y_k)$, where $I_k$, $J_k \subseteq [n]$ are subsets of the available pure strategies and $\Supp(x_k) \subseteq I_k$, $\Supp(y_k) \subseteq J_k$.
\item Rationalization of the data set as $\epsilon$-approximate Nash equilibria.
\end{itemize}

For simplicity of notation, we give the linear program for the basic case of two players; this can easily be extended to any of the above cases. The variables in the linear program are the entries of the payoff matrices $A$ and $B$, and the payoffs $\pi_k$, $\pi_k'$ obtained by the players for each observation $(x_k, y_k)$ in the data set.  


\begin{proposition} \label{p.lp}
A data set $\D=\{(x_k,y_k) \in \Delta^n \times \Delta^n \mid 1 \leq k \leq m\}$ is rationalizable iff the optimal value of (LP) is strictly greater than zero.
\[
\begin{array}{rll}
\textrm{\emph{maximize} } \ & \delta \tag{LP}\\
\textrm{\emph{subject to } } \  &  (A y_k)_i = \pi_k &  \forall k, \forall i \in \Supp(x_k) \\
& (Ay_k )_j \leq \pi_k - \delta & \forall k, \forall j  \notin \Supp(x_k)  \\
& ( x_k^T B )_i = \pi'_k & \forall k, \forall i \in \Supp(y_k) \\
& (x_k^T B )_j \leq \pi'_k - \delta & \forall k, \forall j  \notin \Supp(y_k)  \\
& 0 \leq A_{i,j}, B_{i,j} \leq 1 & \forall i,j \in [n] \\
& \delta \geq 0.
\end{array}
\]
\end{proposition}

\begin{proof}
If the optimal value of the linear program is strictly greater than zero then we have matrices $A$ and $B$ (corresponding to an optimal solution) under which the observations $(x_k,y_k)$ are strict Nash equilibrium. Hence $\D$ is rationalizable.

On the other hand, if $\D$ is rationalizable, say via game $(\hat{A}, \hat{B})$, then we can scale the entries of $\hat{A}$ and $\hat{B}$ (by a large enough positive constant) and add a fixed number to all of them to obtain a (normalized) game $(A,B)$ with entries between $0$ and $1$. Since such an affine transformation preserves the set of strict Nash equilibria, $\D$ is also rationalized by $(A,B)$. In other words, $(A,B)$ gives us a feasible solution to the linear program with objective function value strictly greater than zero. This establishes the claim.
\end{proof}

\section{Main Results: Rationalizations with Low Player Rank}
\label{s.mainresults}

This section includes the main results of the paper, which identify structural properties of data sets that guarantee the existence of rationalizations for which the observed Nash equilibria can be computed efficiently.  To ensure that the observed equilibria can be computed efficiently in the rationalizations we construct rationalizations with low player rank and make use of the following result from \cite{GargJM11}.



\begin{theorem}[\hspace*{-4pt}\cite{GargJM11}]
If the player rank of a bimatrix game is $k$ then all extreme Nash equilibria can be computed in time $O(n^{O(k)})$.
\end{theorem}

The algorithm implicit in this theorem computes all extreme
equilibria, and so it  can be used to compute the observations in the data
set rather than simply some arbitrary equilibria.  This fact is
crucial to the exercise, and provides strong motivation for a focus on player rank, since the goal is to explain the specific
observations in the data set.


\subsection{Observations from a Low Dimensional Subspace}

The first property we connect to player rank is the dimensionality of
the observed strategies in the data set. Given a finite set $S
\subset \mathbb{R}^n$ of $m$ vectors $s_1, \ldots, s_m$, write
$\dm(S)$ to denote the maximum number of linearly independent vectors
in $S$.  Observed strategies that form a low dimensional subspace are
natural candidates for low player rank rationalizations and,
the following theorem shows that --- independent of the size of the data
set --- if the observations form a low dimensional subspace then they
can be rationalized by a game of low player rank.

\begin{theorem}
\label{thm:ub-count}
If a data set $\D$ is rationalizable then it can be rationalized by a game of player rank at most $ \min\{ \dm( \Ob_r ), \dm(\Ob_c) \}$.
\end{theorem}

An immediate consequence of the above theorem is that, if a data set $\D$ is rationalizable, then it can be rationalized by a game of player rank at most $|\D|$. Additionally, later, in Theorem \ref{thm:lb}, we prove a lower bound that highlights that this result is tight.

Importantly, Theorem~\ref{thm:ub-count} has rationalizability as one of its hypotheses, and thus implies that, for data with low-dimensional strategies, the computational constraints have no added empirical or observational content. In other words, any low dimensional data set that is rationalizable without computational constraints is also rationalizable with them.

To prove Theorem \ref{thm:ub-count}, a key technical piece is the following lemma about reconstructing the product of an arbitrary matrix and a low-rank matrix. 

\begin{lemma} \label{lem:lowrank}
Suppose $Y \in \mathbb{R}^{n \times m}$ is a matrix of rank $t$. Then for every matrix $\hat{A} \in \mathbb{R}^{n \times n}$, there exists a matrix $A \in \mathbb{R}^{n \times n}$ of rank at most $t$ that satisfies $A Y = \hat{A} Y$.

\end{lemma}

\begin{proof}
Let $\{ y_1, y_2, \ldots, y_t \}$ be a set of linearly independent columns in $Y$, and define $\hat{Y}$ as the $n \times t$ matrix $[y_1 ~ y_2 ~ \ldots ~ y_t]$. Since each of the columns of $\hat{Y}$ are linearly independent, we can obtain a matrix $\Gamma$ of size $n \times t$ that satisfies $\Gamma^T \, \hat{Y} = I_t$. Thus if we denote the $i$th column of $\Gamma$ by $\gamma_i$, then for all $j \in [t]$,

\begin{align*}
\gamma_j^T y_j & = 1 \qquad \textrm{ and} \\
\gamma_j^T y_k & = 0 \qquad \forall k \in [t] \setminus \{j\} \, .
\end{align*}

\noindent We define matrix $A$ as the following sum of $t$ rank-$1$ outer products:

\begin{align*}
A & = \sum_{j=1}^t \hat{A} y_j \gamma_j^T \, .
\end{align*}

\noindent By construction, the rank of $A$ is at most $t$. Further, for all $y_j$ for $j \in [t]$:

\begin{align}
\label{eq:no-hat}
A \, y_j & = \sum_{k=1}^t \hat{A} y_k \gamma_k^T y_j ~ = ~ \hat{A} \, y_j \, .
\end{align}

Since any column $y$ of matrix $Y$ can be expressed as a linear combination of the columns of $\hat{Y}$, we can write $y = \sum_{j=1}^t \lambda_j y_j$. Then $\hat{A} \, y_i$ is given by

\begin{align*}
\hat{A} \, y  & = \hat{A} \left( \sum_{j=1}^t \lambda_j y_j \right)  =  \sum_{j=1}^t \lambda_j \hat{A}y_j = \sum_{j=1}^t \lambda_j Ay_j  \\
&  = A \left( \sum_{j=1}^t \lambda_j y_j \right) = Ay \, .
\end{align*}

\noindent where the third equality is obtained from~(\ref{eq:no-hat}). Overall we have the desired claim, $AY = A\hat{Y}$. 

\end{proof}


Using this lemma, Theorem \ref{thm:ub-count} can be established as follows.

\begin{proof}[Proof of Theorem~\ref{thm:ub-count}]
Consider a game $(\hat{A}, \hat{B})$ that rationalizes the data set $\D=\{ (x_k,y_k) \}_{k=1}^m$. Write $X$ ($Y$) to denote the matrix whose $k$th row (column) is equal to $x_k$ ($y_k$), for all $k \in [m]$.  Note that $\rank(X) = \dm(\Ob_r)$ and $\rank(Y) = \dm(\Ob_c)$. By Lemma~\ref{lem:lowrank}, there exist matrices $A$, $B$ so that $\rank(A) \le \rank(Y)$, $\rank(B) \le \rank(X)$, and $AY = \hat{A}Y$, $XB = X \hat{B}$. Then $(A,B)$ is the rationalization required by the theorem. We have already shown that $A$ and $B$ are of the required rank. To see that $(A,B)$ rationalize $\D$, note that since $AY = \hat{A}Y$, for all $(x,y) \in \D$ we have $\BR_r^A(y) = \BR_r^{\hat{A}}(y)$. Hence $\Supp(x)  = \BR_r^A(y)$. Similarly, since $XB = X \hat{B}$, $\Supp(y) = \BR_c^B(x)$. Hence $(x,y)$ is a strict Nash equilibrium in $(A,B)$. 
\end{proof}

\subsection{Observations with Small Support Size}

The second structural property of the data set we consider is the support size
of the observations. In spirit, the following theorem complements the
result of Lipton et al.~\cite{LMM} wherein they establish that if the
rank of both the payoff matrices is low then the game contains a
small-support equilibrium. The following result highlights that there is a
connection in the other direction as well.

\begin{theorem}
\label{thm:ss}
Let $\D=\{(x_k,y_k) \in \Delta^n \times \Delta^n \mid 1 \leq k \leq m\}$ be a data set in which $|\Supp(x_k)| \leq s$ for all $k \in [m]$ or $|\Supp(y_k)| \leq s$ for all $k\in [m]$. If the observed strategies $\Ob_r(\D)$ and $\Ob_c(\D)$ are generic then $\D$ can be rationalized by a game with player rank $\leq 2s + 1$.
\end{theorem}

Note that, later, Theorem \ref{thm:lb} highlights that the bound in Theorem \ref{thm:ss} is tight to within a factor of 2.

Our proof of Theorem \ref{thm:ss} uses a construction based on polynomials, and the following lemma is a key technical piece in the argument.  The lemma states that if for all $j \in [m]$, the $j$th column of an $n \times m$ matrix $M$ is obtained by evaluating a degree $d$ polynomial $p_j$ at $1,2,\ldots n$ (i.e., the $j$th column of $M$ is equal to $(p_j(1), p_j(2), \ldots, p_j(n))^T$), then the rank of $M$ is at most $d+1$.
\begin{lemma}
\label{prop:polymat}
Let $p_1, p_2, \ldots, p_m$ be $m$ univariate polynomials over $\mathbb{R}$, and suppose that the degree of each of them is at most $d$. If the $(i,j)$th entry of an $n \times m$ matrix $M$ is equal to $p_j(i)$, for all $i \in [n]$ and $j \in [m]$, then the rank of $M$ is at most $d+1$.
\end{lemma}

\begin{proof}
Write $p_j(x) = a_d^{(j)} x^d + a_{d-1}^{(j)} x^{d-1} + \ldots + a_1^{(j)} x + a_0$ for all $j \in [m]$. $M$ can be expressed as a sum of $d+1$ outer products:
\begin{align*}
M & = \begin{pmatrix} 1^d \\ 2^d \\ \vdots  \\ n^d \end{pmatrix} \begin{pmatrix} a_d^{(1)} & a_d^{(2)} & \cdots & a_d^{(m)} \end{pmatrix} + \begin{pmatrix} 1^{d-1} \\ 2^{d-1} \\ \vdots  \\ n^{d-1} \end{pmatrix} \begin{pmatrix} a_{d-1}^{(1)} & a_{d-1}^{(2)} & \cdots & a_{d-1}^{(m)} \end{pmatrix} + \ldots \\  & \ \    \ldots + \begin{pmatrix} 1 \\ 1 \\ \vdots  \\ 1 \end{pmatrix} \begin{pmatrix} a_0^{(1)} & a_0^{(2)} & \cdots & a_0^{(m)} \end{pmatrix}.
\end{align*}

Note that the rank of an outer product is one and the rank of the sum of two matrices satisfies $\rank(X + Y ) \leq \rank(X) + \rank(Y)$.  Hence the rank of $M$ is no more than $d+1$. 
\end{proof}

Using this lemma, Theorem \ref{thm:ss} can be established as follows.

\begin{proof}[Proof of Theorem~\ref{thm:ss}]
We prove the claim for the case in which the mixed strategies of the row player in the data set $\D=\{(x_k,y_k) \in \Delta^n \times \Delta^n \mid 1 \leq k \leq m \}$ are of support size at most $s$. A construction similar to the one presented below takes care of the alternate case wherein $|\Supp(y_k)| \leq s$ for each $k \in [m]$.

We consider a polynomial $p_k$ that satisfies $\argmax_x p_k(x) = \Supp(x_k)$ and has degree $2 |\Supp(x_k)|$. In particular,
\begin{align*}
p_k(x) & := - \prod_{i \in \Supp(x_k) } (x - i)^2.
\end{align*}

Say $\Supp(x_k) = \{i_1, i_2, \ldots , i_s\} \subset [n]$, then the polynomial $p_k$ vanishes exactly at $i_1, i_2, \ldots, i_s$ and is negative elsewhere. Hence,  $\Supp(x_k)$ is the set of points at which $p_k$ attains its maximum value. In addition, the degree of $p_k$ is $2 |\Supp(x_k)|$.

Consider the $n \times m$ matrix $P$ in which the $k$th column is equal to $(p_k(1), p_k(2),$ ...$,p_k(n))^T$. By construction, for all $k \in [m]$, degree of the polynomial $p_k$ is no more than $2s$. Therefore, Lemma~\ref{prop:polymat} implies that the rank of $P$ is at most $2s+1$. Moreover, the set of the largest components of the $k$th column of $P$ (i.e., $\argmax_i P_{i,k}$) is exactly equal to $\Supp(x_k)$. Since, $\argmax_{i \in [n]} P_{i,k}  = \argmax_{i \in [n]} $  $p_k(i) \ = \Supp(x_k).$



Recall that the mixed strategies in $\Ob_c(\D)$ are generic. Therefore, we can find an $m \times n$ matrix $V$ that satisfies the following equality for all $y_k \in \Ob_c( \D)$: $V y_k = e_k^{(m)}$. Here $e_k^{(m)}$ is the $m$-dimensional vector with a $1$ in the $k$th coordinate and $0$s elsewhere.

Set the payoff matrix of the row player $A = PV$. Rank of the product of two matrices satisfies: $\rank(XY) \leq \min \{ \rank(X), \rank(Y) \}$. Hence $\rank(A) \leq 2s+1$.

For all $(x_k,y_k) \in \D$, we have $A y_k = P e_k^{(m)} =(p_k(1), p_k(2), \ldots, p_k(n))^T$. Hence, the set of the largest components of the vector $A y_k$ is equal to $\Supp(x_k)$. Overall, under the payoff matrix $A$, we have $\BR_r(y_k) = \Supp(x_k)$, for all $(x_k, y_k) \in \D$. That is, $A$ rationalizes the mixed strategies of the row player. 
\end{proof}

\subsection{Observations with Low Chromatic Number}

The third, and final, structural property of data sets that we consider is the chromatic number.  Intuitively, the chromatic number quantifies the degree of intersection between the observed mixed strategies, and hence it is a relevant measure of the structural complexity of data.

For a data set $\D$, we define the row chromatic number $\rcn(\D)$
and the column chromatic number $\ccn(\D)$ as the chromatic numbers of
graphs $G_r$ and $G_c$, defined as follows. For the row
chromatic number, $\rcn(\D)$, construct graph $G_r$ with a vertex
corresponding to each observation in $\Ob_r$. For distinct observations $(x,y)$
and $(x',y')$ in $\D$, if $\Supp(x) \cap \Supp(x') \neq \emptyset$
then the graph $G_r$ has an edge between the corresponding vertices. Then
set $\rcn(\D) = \chi(G_r)$, i.e., the chromatic number of graph
$G_r$. The column chromatic number is defined similarly using
intersections $\Supp(y) \cap \Supp(y')$. The chromatic number of the
data set, $\cn(\D)$, is defined to be the minimum of $\rcn(\D)$ and
$\ccn(\D)$.

\begin{theorem}
\label{thm:cn}
Let $\D$ be a data set with chromatic number equal to $\cn(\D)$. If the observed mixed-strategy sets $\Ob_r(\D)$ and $\Ob_c(\D)$ are generic then $\D$ can be rationalized by a game of player rank at most $2\cn(\D)$.
\end{theorem}

Note that, later, Theorem \ref{thm:lb} highlights that the bound in Theorem \ref{thm:ss} is tight to within a factor of 2.

Importantly, like the case of support size, the bound on the player rank in Theorem \ref{thm:cn} is not exactly dependent on the size of the data set, but rather only on the ``richness'' of the observations in terms of the structure of the underlying graph.

Of course, in general the chromatic number of a graph is hard to compute. However, an easy upper bound is the maximum degree of any vertex in $G_r$ and $G_c$ plus one, which then can be interpreted as follows: $ \rcn(\D)  \le \max_{(x,y) \in \D} |\{(x',y') \in \D: \Supp(x) \cap \Supp(x') \neq \phi \}| $ and
$ \ccn(\D) \le \max_{(x,y) \in \D} $  $|\{(x',y') \in \D: \Supp(y) \cap \Supp(y') \neq \phi \}|$. Though these bounds provide intuition, it is obvious the chromatic numbers can be much less than these upper bounds, e.g., if the graph $G_r$ is a star.

The proof of Theorem \ref{thm:cn} starts by focusing on data sets where the row chromatic number is one, i.e., the supports of all the observations for the row player are pairwise disjoint.  In this case we show, in the following Lemma, that the data set can be rationalized by a game with row player rank 2. We start with a proof of the initial lemma. 

\begin{lemma}
\label{lem:cn}
Let $\D'$ be a data set with row chromatic number $\rcn(\D') = 1$. If the set of observations $\Ob_c(\D')$ is generic then there exists a rank $2$ matrix $A'$ that rationalizes the row player's strategies in $\D'$.
\end{lemma}

\begin{proof}
Say $\D'$ consists of $m$ observations $(x_1,y_1), \ldots, (x_m, y_m)$. Since the support sets of the strategies of the row player are disjoint, any (pure) strategy $i \in [n] $ of the row player is in at most one such support set. For pure strategy $i \in [n]$, let $\sigma(i) := k$ such that $i \in \Supp(x_k)$. If row $i$ is not in the support of any strategy, let $\sigma(i) = m+1$.

We construct three vectors $u$, $w$ and $f$ in $\mathbb{R}^n$. Values are assigned to the components of the vectors $u$ and $w$ directly: for $i \in [n]$, entry $u_i = -\sigma(i)^2$, and $w_i = \sigma(i)$. We define vector $f$ so that for all $y_k \in \Ob_c(\D')$, $f^T y_k = k$. Since the set of observations $\Ob_c(\D')$ are generic, such a vector $f$ exists and can be obtained. Matrix $A'$ is then defined as follows: $A' := u \mathbf{1}_n^T + 2 w f^T$, where $\mathbf{1}_n$ is the $n$-column vector consisting of all 1's.

Since $A'$ is the sum of two outer products, it has rank $2$. We show that with payoffs from $A'$ the strict Nash requirement for the row player, $\Supp(x_k) = \BR_r(y_k)$, is satisfied for all $(x_k,y_k) \in \D'$.  Hence, we get the desired lemma.

For any $y_k \in \Ob_c(\D')$ , by our construction, $A' y_k = u + 2 w f^T y_k = u + 2 k w$. For a fixed $k$, consider the $j$th component of the vector $A' y_k$, which by the construction is $- \sigma(j)^2 + 2k \sigma(j)$. Note that this expression is maximized when $\sigma(j) =k $ and is strictly less for other values of $\sigma(j)$. Further, by construction, if $\sigma(j) = k$ then $j \in \Supp(x_k)$. Thus, the maximum components of $A' y_k$ are exactly those that correspond to the support of $x_k$, and hence under payoff matrix $A'$, $\BR_r(y_k) = \Supp(x_k)$. 
\end{proof}

We show now how this lemma can then be used to construct games with low player rank for data sets with larger chromatic numbers.

\begin{proof}[Proof of Theorem \ref{thm:cn}]
We constructively show that there exists a matrix $A$ that rationalizes the mixed strategies of the row player in $\D$ and has rank no more than $2 \rcn(\D)$. Similarly, we can construct a payoff matrix $B$ of rank at most $2\ccn(\D)$ for the column player. This establishes the existence of the game $(A,B)$ that rationalizes $\D$ with the required player rank.

Let $t$ be the row chromatic number of the data set, i.e., $t = \rcn(\D)$, and let $\chi(G_r)$ be the graph coloring that defines the row chromatic number. We partition the observations in data set $\D$ into $t$ sets $\D_1, \D_2, \ldots, \D_t$ according to the color assigned to the vertex corresponding to each observation, so that observations with the same color are in the same partition. Note that by this technique, for all observations within the same partition, the supports of the observations for the row player are disjoint. 

By Lemma~\ref{lem:cn}, for each data set $\D_j$, we can obtain a rank-2 matrix $A_j$ so that $\BR_r(y_i) = \Supp(x_i)$ with $A_j$ as the payoff matrix. In order to combine these matrices, for $j \in [t]$, we define matrix $V_j$ as satisfying the following property:

\begin{align*}
V_j y & = y \textrm{ for $y \in \Ob_c(\D_j)$, and } \\
V_j y & = \mathbf{0} \textrm{ otherwise. }
\end{align*}

\noindent Since the set of strategies of the players are generic, we can obtain such matrices. We then define the payoff matrix $A$ as
$
A = \sum_{j=1}^t A_j V_j \, .
$
Since each $A_j$ is of rank 2, matrix $A$ is of rank $2t = 2 \rcn(\D')$. To see that $A$ rationalizes the data set $\D$, note that for $j \in [t]$ and $(x,y) \in \D_k$,
$
Ay = \sum_{j=1}^t A_j V_j y = A_k y,
$
and by construction of $A_k$, $\BR_r(y) = \Supp(x)$. 
\end{proof}

\subsection{A Unifying Result}

The previous sections have identified three structural properties for data sets that ensure the existence of low player rank rationalizations.  In this section, we present a unifying result that extends the previous three theorems to provide a more robust low-rank construction and, in particular, shows that addition of a small number of observations to a data set does not have a big impact on the player rank necessary to rationalize the data. Specifically, we establish low-rank rationalizations for data sets that can be partitioned into three sets which are structurally simple in terms of dimensionality, support size, and chromatic number, respectively. For example, say we have a data set $\D$ in which all but $t$ observed mixed strategies are of support size $s$, then it can be partitioned into a set that has support bounded by $s$ and a set that has dimensionality bounded by $t$ (and an empty set that has chromatic number zero).  The following corollary then shows how to construct a rationalization for $\D$ of player rank at most $(2s +1) + t$.\footnote{The precise bound from the corollary is $2(s+t) +1$, but this can be strengthened to $(2s +1) + t$. For ease of presentation, we present the corollary with slightly loose factors.}

To obtain such a generalization we introduce the notion of the \emph{composite number} of a data set, which considers a 3-partition of the data set and combines the dimensionality of the first partition, the support size of the second partition, and the chromatic number of the third partition.

\begin{definition}[Composite number]
The row composite number $\comp_r(\D)$ of a data set $\D$ is defined to be the smallest number for which  there exists a $3$-partition of $\D$, $\{ \D_1, \D_2, \D_3 \}$, that satisfies $\dim(\Ob_c(\D_1)) + \max_{(x,y) \in \D_2} |\Supp(x)| + \rcn(\D_3) = \comp_r(\D)$. The column composite number $\comp_c(\D)$ is defined similarly.

The composite number of a data set, $\comp(\D)$, is the minimum of the row and column composite number: $\comp(\D) :=\min \{ \comp_r( \D), \comp_c(\D)\}$.

\end{definition}

\begin{corollary}
\label{cor:comp}
Let $\D$ be a data set with composite number $\comp(\D)$. If the observed mixed-strategy sets $\Ob_r(\D)$ and $\Ob_c(\D)$ are generic then $\D$ can be rationalized by a game of player rank at most $2 \comp(\D) + 1$.
\end{corollary}

\begin{proof}
Below we show that there exists a payoff matrix $A$ of rank at most $2\comp_r(\D) +1 $ that rationalizes the row player's strategies in $\D$. A similar argument establishes the existence of a matrix $B$ (which rationalizes the column player's strategies) of rank no more than $2 \comp_c(\D) + 1$, and hence we get the desired claim.

Say $\{ \D_1, \D_2, \D_3\}$ is a $3$-partition that satisfies $\dim(\Ob_c(\D_1)) + \max_{(x,y) \in \D_2} |\Supp(x)| + \rcn(\D_3) = \comp_r(\D)$. The constructions of Theorems~\ref{thm:ub-count},~\ref{thm:ss}, and~\ref{thm:cn}, imply that there exist matrices $A_1$, $A_2$, $A_3$ of rank $\dim(\Ob_c(\D_1))$, $2 \max_{(x,y) \in \D_2} |\Supp(x)| +1 $, and $2 \rcn(\D_3) $ respectively such that $A_i$ rationalizes the row player's observations in $\D_i$ for all $i \in [3]$.

Since $\Ob_c(\D)$ is generic, there exists matrix $V_i$ for all $i \in [3]$ that satisfies the following equalities:
\begin{align*}
V_i y & = y \qquad \forall \ y \in \Ob_c(\D_i)   \\
V_i y & = \mathbf{0} \qquad \forall y \in  \Ob_c(\D) \setminus  \Ob_c(\D_i)
\end{align*}
Note that payoff matrix $A = \sum_{i=1}^3 A_i V_i$ rationalizes the row players observations in $\D$ and is of rank at most $2 \comp_r(D) + 1$. 
\end{proof}

This result serves as another illustration of why player rank is an appealing choice for the revealed preference exercise in this paper, since it allows us to merge the constructions used in Theorems \ref{thm:ub-count}, \ref{thm:ss}, and \ref{thm:cn}. 

\section{A Lower Bound on Player Rank}
\label{s.lb}

The results to this point of the paper have focused on constructing rationalizing games with low player rank, thus guaranteeing the observed equilibria can be computed efficiently.  It is also natural to ask if there exist data sets that \emph{require} rationalizations to have large player rank.  In the following, we show that such data sets do exist. In particular,  there exists a data set that requires any rationalization to have player rank at least $n-1$.   

\begin{theorem}
\label{thm:lb}
Any game $(A,B)$ that rationalizes $\D=\{ (u_k, u_k) \mid k \in \{2,3,\ldots, n \} \}$  has player rank at least $n-1$, i.e., $\rank(A) \geq n-1$ and $\rank(B) \geq n-1$.
\end{theorem}

\begin{proof}
For $2\leq k \leq n$, recall that $u_k \in \Delta^n$ denotes the uniform distribution over the set $\{1,2, \ldots, k \}$. Consider the following data set with $n-1$ observations, $\D=\{ (u_k, u_k) \mid k \in \{2,3,\ldots, n \} \}$.

Note that the data set $\D$ is rationalizable. In particular, the game obtained by setting the payoff matrices of both the players to $I_n$ (the $n \times n$ identity matrix) rationalizes $\D$. The player rank of the rationalization, $(I_n, I_n)$, is $n$. Below we establish that in fact the player rank of any game that rationalizes $\D$ is at least $n-1$.

Say $(A,B)$ is a rationalization of $\D$. Let $A_{(i)}$ be the $i$th row of the matrix $A$ and for $2 \leq j \leq n$, we define $n-1$ vectors $v_{j}$s as follows: $v_{j} := A_{(1)} - A_{(j)}$.  Note that $v_{j}$s lie in the row space of $A$. We will show that $v_{j}$s are linearly independent and hence get that the dimension of the row space of $A$ is at least $n-1$. This, in turn, proves that the rank of $A$ is at least $n-1$. Since the data set $\D$ is symmetric, via a similar argument, we can establish that the rank of $B$ is no less than $n-1$. This overall establishes the stated claim that the player rank of $(A,B)$ is at least $n-1$.

Since mixed strategy pair $(u_k,u_k)$ is a strict Nash equilibrium in $(A,B)$, we have $e_1^T A u_k = e_j^T A u_k$ for all $j \in \{2,\ldots, k\}$ and $e_1^T Au_k > e_j^T A u_k$ for all $j \in \{k+1,\ldots, n\}$. That is, $A_{(1)}^T u_k = A_{(j)}^T u_k $ for all $j \in [k]$ and $A_{(1)}^T u_k > A_{(j)}^T u_k $ for all $j \notin \{k+1,\ldots, n\}$. We can rewrite these equalities and inequalities using the definition of $v_j$s as follows:
\begin{align*}
v_j^T u_k & = 0 \qquad \forall j \in \{2,\ldots,k\} \textrm{ and} \\
v_{j}^T u_k &> 0 \qquad \forall j \in \{k+1, \ldots, n \}.
\end{align*}

Hence, for all $2 \leq k < n $, vector $v_2, v_3, \ldots, v_{k+1}$ are linearly independent. Say for contradiction that they are linearly dependent. Then we can write $v_{k+1}$ as a linear combination of $v_2, \ldots, v_k$, i.e., $v_{k+1} = \sum_{j=2}^k \lambda_j v_j$. Taking inner product of both sides of this equation with $u_k$ leads to a contradiction.

Overall, we get that the vectors $v_2, \ldots, v_n$ are linearly independent and this completes the proof. 
\end{proof}

This result is important for two reasons.  First, the theorem highlights that Theorems \ref{thm:ub-count}, \ref{thm:ss}, and \ref{thm:cn} are (nearly) tight.  Specifically, by construction, data set $\D'$ satisfies: (i) the observed strategies of each player lie in a subspace of dimension $n-1$; (ii) each observed strategy has support size at most $n$; and (iii) the chromatic number of the data set is $n-1$. It follows immediately that the bounds in Theorem~\ref{thm:ub-count} are exactly tight, and the bounds in Theorems~\ref{thm:ss} and~\ref{thm:cn} are tight to within a factor of 2.

Second, the lower bound strongly suggests that adding computational constraints to the theory of Nash equilibrium has testable implications, i.e., it is likely that there exist data sets for which any payoff matrices that explain the data sets require players to solve computationally hard problems when computing the observed equilibrium.  This is in contrast with single-person consumer theory \cite{consumer}. It is still possible that rationalizing games could be simple, but it seems unlikely. Investigating this issue further is an intriguing direction for future work.

\section{Discussion and Future Work}
\label{s.end}


Our work is the first to consider the implications of computational complexity for mixed strategy behavior in games with multiple players. For our results we make certain simplifying assumptions that are natural given prior work on revealed preference theory for a single agent and on the computational complexity of equilibrium. In this section we discuss these assumptions, their relaxation, and interesting open problems.

\paragraph{Observations of exact mixed strategy behavior.}
We assume that exact mixed strategy behavior is observable, and ignore the presence of noise in our observations. This assumption is justified as our results would be easier to establish under the extra flexibility afforded by noisy observations.  Further, while a number of other observations can be considered as possible inputs, e.g., pure strategy samples from a mixed Nash equilibrium, the model of observations we consider for mixed strategies is the natural first step.  The assumption that exact mixed strategy behavior can be observed is also made in the theory of individual stochastic choice, e.g.,
\cite{luce1959,mcfadden_zarembka,mcfad91,mcfadden2005revealed}. In the theory of
individual stochastic choice, these models have developed into actual
empirical tools that are very heavily used among economists (for
example, they are used as a standard tool by most empirical economists). See
\cite{anderson1992discrete} for an exposition of the theory as used by
empiricists. 

\paragraph{Repeated observations from the same game.} The data sets considered in this paper comprise of multiple observations of mixed-strategy behavior that potentially correspond to different equilibria.  Superficially, this may suggest that one could consider observed player behavior as resulting from a learning dynamic, or as outcomes from a repeated game.  However, given the ``anything goes'' message of the folk theorem (see~\cite{fudenberg1991game}), studies in experimental economics are often designed explicitly to avoid repeated game behavior. This is typically enforced through anonymous and random matching experimental designs. Thus, Nash equilibria is often the model of interest.  Further, considering repeated games or learning dynamics would again involve relaxing some of the constraints considered in the current paper, since it would no longer be required that the payoff be maximized at each observed strategy profile. It is likely that these relaxations would make it easier to provide explanation of the data via tractable payoff matrices. 

\paragraph{Alternate notions of tractability.} Our results provide strong motivation for the use of player rank as a notion of tractability; however there are many other properties which ensure the existence of efficient algorithms for computing Nash equilibria.  These include: a \emph{game rank} of zero or one,\footnote{If the game rank is 3 or larger then the computation of a Nash equilibria is PPAD-hard~\cite{rank4}.} where the game rank is the rank of $C:=A+B$; the existence of a
potential function; or the existence of a pure Nash equilibrium. None of these properties is appropriate for use in the exercise here because each is binary: either a data set possesses the particular
property, or it does not; and  absence of the property renders
algorithms based on the property useless in computing equilibria. Further, in Appendix \ref{s.proofs1}, we show that for each property other than
player rank, simple data sets with a small number of observations
necessitate rationalizations that do not satisfy the property.

\paragraph{Genericity of the observations.} Theorems \ref{thm:ss} and \ref{thm:cn} assume that the set of observations are generic. For games with $n$ pure strategies for the players, if the number of observations are less than $n$, this is a mild assumption, since a tiny perturbation is sufficient to ensure genericity. The assumption does necessitate that the number of observations be at most $n$. Since computational complexity is interesting for large values of $n$, our work is still relevant for large classes of data. From a technical point of view however, it would be very interesting to see if this assumption could be removed; or if our computational complexity results could be extended to other classes of data sets.

\paragraph{Learning versus testing.}  Our focus in this paper is on testing. We do not address the related problem of inference in this paper. The problem of actually learning the underlying payoffs of players by observing player behavior has previously been studied for the case of a single consumer~\cite{BeigmanV06,ZadimoghaddamR12} and for correlated equilibrium with multiple players~\cite{WaughZB11}. Our work, while primarily focused on testing, does offer some insight into how structural properties of the observed behavior affect the rank of the payoff matrices we are trying to learn. A more comprehensive study of learning the underlying utilities of players in game-theoretic settings  is an important direction for future work.

\section*{Acknowledgments}
This research was supported by NSF grants CNS-0846025, EPAS-1307794, and CCF-1101470, along with a Linde/SISL postdoctoral fellowship.

\bibliographystyle{plain}
\bibliography{rev-nash}

\appendix

\section{Game Rank, Potential Games, and Pure Strategy Equilibria}
\label{s.proofs1}

The goal of this paper is to understand when it is possible to rationalize data via payoff matrices for which the mixed strategies observed are efficiently computable. Given the hardness of computing equilibria in general, this requires that the rationalizations we generate must have some special property that allows for efficient computation. We do this by focusing on rationalizations with small player rank; however there are a number of other properties that could be considered.  For example, a small game rank, the existence of a potential function, or the existence of a pure Nash equilibrium.  In the following, we highlight that these alternatives are not well-suited for use in this paper.

\subsection{Game Rank}

The connection between game rank and computational efficiency has only recently begun to be understood.  To this point, polynomial-time algorithms to compute Nash equilibria are known when the game rank is either zero \cite{Nisan2007} or one~\cite{AdsulGMS11}, and it has recently been shown that when the game rank is four or more computing an equilibrium is PPAD-hard~\cite{rank4}.  Though incomplete, these results are already problematic for the use of game rank in this paper.  In particular, the following result highlights that only very small data sets can be guaranteed to have game rank small enough to ensure that a computationally efficient algorithm exists, e.g., there is a data set with 9 observations that necessitates game rank of at least two.

\begin{theorem}
\label{thm:lbC}
There exists a rationalizable data set $\D$ with $2n+1$ observations such that any game $(A,B)$  that rationalizes $\D$ has game rank at least $n-2$, i.e., $\rank(A+B) \geq n-2$.
\end{theorem}

\begin{proof}
Let $u_n \in \Delta^n$ be the uniform distribution over $[n]$ and $e_k \in \Delta^n$ be the vector with a $1$ in the $k$th coordinate and $0$'s elsewhere. Write $v_k$ to denote the uniform distribution over $[n] \setminus \{ k \}$. We consider the following data set with $2n+1$ observations, $\D
= \{(e_k,e_k) \mid 1 \leq k \leq n  \} \cup \{ (v_k,v_k)  \mid 1 \leq k \leq n \} \cup  \{ (u_n,u_n) \} $. Note that $\D$ can be rationalized by the game $(I_n, I_n)$, where $I_n$ is the $n \times n$ identity matrix.

Say game $(A,B)$ rationalizes $\D$. First we show that in every column of $A$ all the off-diagonal entries are equal to each other. That is, for all $k \in [n]$ and for all $i, i' \in [n] \setminus \{ k \}$ we have $A_{i,k} = A_{i',k}$. A similar result holds for the rows of matrix $B$.

Since $(u_n, u_n)$ is a strict Nash equilibrium in $(A,B)$ we have $\Supp (u_n) = \BR_r(u_n) $. This implies that all the components of the vector $A u_n$ are equal, i.e., the row sums of $A$ are equal to each other. Formally,

\begin{align}
\label{eq:rs}
\sum_{j} A_{i,j} = \sum_{j} A_{i',j} \quad \forall i,i' \in [n].
\end{align}

Similarly, the fact that $(v_k, v_k)$ is a strict Nash equilibrium implies $\Supp (v_k) = \BR_r(v_k) $. In particular, for all $i, i' \in \Supp(v_k)$ the $i$th and the $i'$th component of $Av_k$ must be equal to each other. Since the $i$th component of the vector $A v_k$ is equal to $\frac{1}{n-1} \sum_{j \neq k } A_{i, j } $ and $\Supp(v_k) = [n] \setminus \{k \}$, we have the following equality for all $i,i' \in [n] \setminus \{k \}$:
\begin{align}
\label{eq:rso}
\sum_{j \neq k} A_{i,j} = \sum_{j \neq k } A_{i',j}.
\end{align}

Subtracting (\ref{eq:rso}) from (\ref{eq:rs}) gives us $A_{i,k} = A_{i',k}$ for $i,i' \in [n] \setminus \{ k \}$.

Finally, using the fact that $(e_k, e_k)  \in \D$  we get that $(k,k)$ is a pure and strict Nash equilibrium in $(A,B)$ for all $k$. Therefore, $A_{k,k} > A_{i, k}$ for all $i \neq k$. Say the off-diagonal entries of the $k$th column of $A$ are equal to $\alpha_k$. We have $A_{k,k} > \alpha_k$ and matrix $A$ has the following form:

\begin{align*}
A = \begin{pmatrix}
A_{1,1} & \alpha_2 & \alpha_3 & \cdots & \alpha_n \\
\alpha_1 & A_{2,2} & \alpha_3 & \cdots & \alpha_n \\
\alpha_1 & \alpha_2 & A_{3,3} & \cdots & \alpha_n \\
\vdots & \vdots & \vdots & \vdots & \vdots  \\
\alpha_1 & \alpha_2 & \alpha_3 & \cdots & A_{n,n}
\end{pmatrix}
\end{align*}
We can write $A$ as the sum of a diagonal matrix and an outer product.
\begin{align*}
A = \begin{pmatrix} A_{1,1} - \alpha_1 \\
& A_{2,2} - \alpha_2  & & \textrm{\huge 0 } \\
& & A_{3,3} - \alpha_3 \\
& \textrm{\huge 0 }  \\
& & & A_{n,n} - \alpha_n
\end{pmatrix}
+ \begin{pmatrix}
1 \\
1 \\
\vdots \\
1
\end{pmatrix}
\begin{pmatrix} \alpha_1 & \alpha_2 & \cdots \alpha_n \end{pmatrix}
\end{align*}
Write $D$ to denote the above diagonal matrix and $P$ to denote the outer product. We have $A = D + P$. Note that all the diagonal entries of $D$ are positive, since $A_{k,k} > \alpha_k$ for all $k$. Similarly, we can decompose column player's payoff matrix $B$ into a diagonal matrix $D'$ and an outer product $P'$, i.e., $B = D' + P'$. Like $D$, the all the diagonal entries of $D'$ are positive.

Overall, we have $ A + B = D + D' + P + P'$. The rank of the sum of two matrices satisfies $\rank(X + Y ) \leq \rank(X) + \rank(Y)$. Therefore, $\rank(D + D') \leq \rank(A+B) + \rank( - (P+P'))$. Since $P$ and $P'$ are outer products, $\rank(-( P + P')) \leq 2$. The diagonal entries of both $D$ and $D'$ are positive, hence matrix $D + D'$ has full rank. This gives us the desired bound, $\rank(A+B) \geq n-2$. 
\end{proof}

\subsection{Potential Games}

When a game has a potential function, it is termed a \emph{potential game}, and an appealing property of such games is that a pure strategy equilibrium is guaranteed to exist (e.g.,~\cite{Nisan2007}).   Not surprisingly, this property is limiting for the purposes of this paper.  That is, if we were to use the existence of a pure strategy equilibria as a property to yield efficient computability of an equilibrium in the rationalizing game, then we would be restricted to extremely limited data sets. To see this, note that there are very simple data sets that cannot be rationalized by a game that has a pure Nash equilibrium, and consequently cannot be rationalized by a potential game.

\begin{theorem}
\label{thm:lbpne}
There exists a rationalizable data set $\D$ with three observations such that any game $(A,B)$ that rationalizes $\D$ does not possess a pure Nash equilibrium.
\end{theorem}

\begin{proof}
We consider a game where each player has 3 strategies, and a data set consisting of the following three observations: $((1,0,0),(0,1/2,1/2))$; $((0,1,0),(1/2,0,$ $1/2))$, and $((0,0,1),(1/2,1/2,0))$. Thus the row player plays a different pure strategy in each observation, while the column player randomizes uniformly over two strategies. To see that any rationalization by matrices $A$, $B$ does not admit a pure Nash equilibrium, suppose for a contradiction that $(i,j)$ is in fact a pure Nash equilibrium and consider the matrix $B$. The data set enforces that the maximum entry in each row is not unique, and hence no entry can be a strict pure Nash equilibrium. 
\end{proof}


\end{document}